\documentclass[draftpaper]{article}

\usepackage{pstricks,calc,pst-node,graphics}
\usepackage{amsgen,amsmath,amssymb,color}
\usepackage[pdftex]{graphicx}
\usepackage{rotating}
\usepackage{lscape}
\usepackage{longtable}
\usepackage{enumitem}
\usepackage{graphicx}
\usepackage{type1cm}
\usepackage{eso-pic}
\usepackage{pst-circ}

\renewcommand{\d}{\mathrm{d}}

\newcommand{\ugualino}{\mbox{\tiny $=$}}

    \newcommand{\newpar}{{\vspace{0.3cm} \noindent}}

\newcommand{\indep}{\mbox{$\perp\!\!\!\perp$}}

\psset{arrows=->,fillcolor=white,fillstyle=solid}
\newcommand{\Show}[1]{\psovalbox{#1}}

 \definecolor{pink}{rgb}{1,
  .75, .8} \definecolor{emcolor}{rgb}{1,0,0}

\newtheorem{theorem}{Theorem}[section]

\newtheorem{definition}{Definition}[section]

\newtheorem{Definition}[definition]{Definition}

\newenvironment{proof}[1][Proof]{\begin{trivlist}
\item[\hskip \labelsep {\bfseries #1}]}{\end{trivlist}}

\makeatletter
\AddToShipoutPicture*{%
\setlength{\@tempdimb}{.5\paperwidth}%
\setlength{\@tempdimc}{.5\paperheight}%
\setlength{\unitlength}{1pt}%
\put(\strip@pt\@tempdimb,\strip@pt\@tempdimc){%
\makebox(0,0){\rotatebox{45}{\textcolor[gray]{0.85}%
{\fontsize{2cm}{1.8cm}\selectfont{}}}}%
}%
}
\makeatother

\begin{document}
\DeclareGraphicsExtensions{.pdf,.png,.gif,.jpg,.tif,.ps}
\title{\textcolor[rgb]{0.00,0.00,0.63}{\huge Deep determinism and the
    assessment of mechanistic interaction between categorical and
    continuous variables\\
\rule{0cm}{0.01cm}}}
\author{CARLO BERZUINI \hspace{0.2cm} and \hspace{0.2cm} A. PHILIP DAWID\\
\rule{0cm}{0.2cm}\\
Statistical Laboratory,\\
Centre for Mathematical Sciences,\\
  University of Cambridge, UK\\
    \rule{0cm}{0.2cm}\\
  \hspace{0.01cm} [c.berzuini,apd]@statslab.cam.ac.uk}
\date{}
\maketitle

\begin{center} {\textcolor[rgb]{0.00,0.00,0.63}{\bf Summary}}
\end{center}

\noindent {\small Our aim is to detect mechanistic interaction between
  the effects of two causal factors on a binary response, as an aid to
  identifying situations where the effects are mediated by a common
  mechanism.  We propose a formalization of mechanistic interaction
  which acknowledges asymmetries of the kind ``factor $A$ interferes
  with factor $B$, but not viceversa''. A class of tests for
  mechanistic interaction is proposed, which works on discrete or
  continuous causal variables, in any combination. Conditions under
  which these tests can be applied under a generic regime of data
  collection, be it interventional or observational, are discussed in
  terms of conditional independence assumptions within the framework
  of Augmented Directed Graphs.  The scientific relevance of the
  method and the practicality of the graphical framework are
  illustrated with the aid of two studies in coronary artery disease.
  Our analysis relies on the ``deep determinism'' assumption that
  there exists some relevant set $V$ --- possibly unobserved --- of
  ``context variables'', such that the response $Y$ is a deterministic
  function of the values of $V$ and of the causal factors of
  interest. Caveats regarding this assumption in real studies are
  discussed.}

\section{\textcolor[rgb]{0.00,0.00,0.63}{\bf Introduction}}
\label{Introduction}

Let the binary random variable $Y$ indicate occurrence ($Y\ugualino1$)
or non-occurrence ($Y\ugualino0$) of an outcome event of interest, and
let $Y$ depend causally (in a sense to be later clarified) on factors
$A$ and $B$.  Also consider a real but possibly unobservable variable
or set of variables $V$, which collude with $A$ and $B$ to cause the
response $Y$, as illustrated by the directed graph of
Figure~\ref{Figure 1}{\em (a)}. In general, even were we to know $A$,
$B$ and $V$, the response $Y$ would not be fully determined, but would
retain an element of random variation.  In certain applications,
however, it might be reasonable to assume that there exists some
relevant set of variables $V$, which we will term {\em context
  variables\/}, such that the binary response $Y$ is {fully
  determined}, without further variation, by $V$ and the values we
impose on $A$ and $B$. More precisely, consider the collection of
(real or hypothetical) interventional regimes where we force $A$ and
$B$ to take on some configuration $(a,b)$. Then the assumption is
that, under such regimes, we have:

\begin{equation}
  \label{eq:f}
  Y = f(A,B,V)
\end{equation}
\noindent for some (typically unknown) function $f$.  Thus, for any
value of $V$, the $(a,b)$ configuration which we force upon $(A,B)$
will precisely dictate whether or not the event $Y=1$ will occur.  We
call this assumption {\em deep determinism\/}.

\begin{figure}
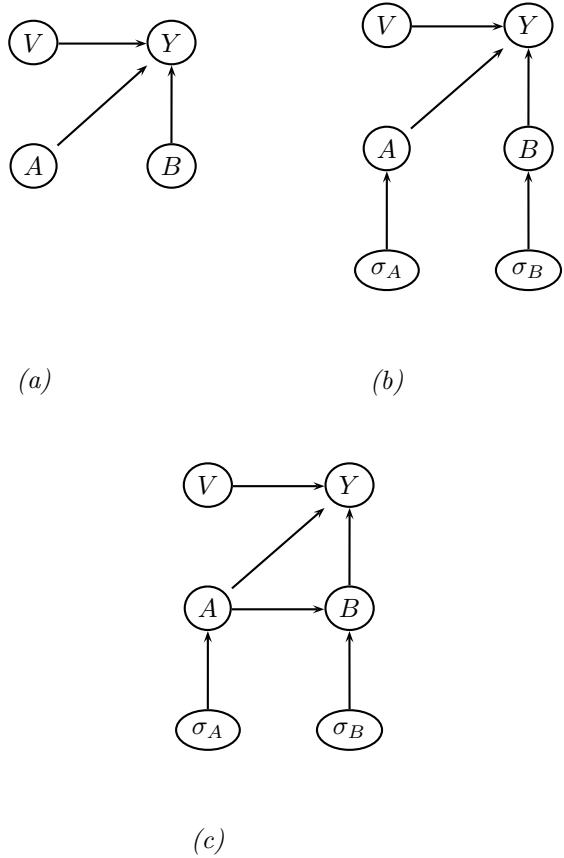

  \begin{center}
    \psframebox[fillstyle=solid,fillcolor=white,
    linestyle=none,framesep=.4,arrowsize=4,arrowlength=4]{%
      \begin{psmatrix}[ref=t,rowsep=1cm,colsep=1cm]
        \Show{$V$}&\Show{$Y$}\\
        \Show{$A$}&\Show{$B$}\\
        \rule{1cm}{0cm}\\
        {\em (a)} \ncLine{->}{1,1}{1,2}\Aput{$$}
        \ncLine{->}{2,1}{1,2}\Aput{$$} \ncline{->}{2,2}{1,2}\Aput{$$}
      \end{psmatrix}} \hspace{1cm}
    \psframebox[fillstyle=solid,fillcolor=white,
    linestyle=none,framesep=.4,arrowsize=4,arrowlength=4]{%
      \begin{psmatrix}[ref=t,rowsep=1cm,colsep=1cm]
        \Show{$V$}&\Show{$Y$}\\
        \Show{$A$}&\Show{$B$}\\
        \Show{$\sigma_A$}&\Show{$\sigma_B$}\\
        {\em (b)} \ncLine{->}{1,1}{1,2}\Aput{$$}
        \ncLine{->}{2,1}{1,2}\Aput{$$} \ncline{->}{2,2}{1,2}\Aput{$$}
        \ncline{->}{3,1}{2,1}\Aput{$$} \ncline{->}{3,2}{2,2}\Aput{$$}
      \end{psmatrix}} \hspace{1cm}
    \psframebox[fillstyle=solid,fillcolor=white,
    linestyle=none,framesep=.4,arrowsize=4,arrowlength=4]{%
      \begin{psmatrix}[ref=t,rowsep=1cm,colsep=1cm]
        \Show{$V$}&\Show{$Y$}\\
        \Show{$A$}&\Show{$B$}\\
        \Show{$\sigma_A$}&\Show{$\sigma_B$}\\
        {\em (c)} \ncLine{->}{1,1}{1,2}\Aput{$$}
        \ncLine{->}{2,1}{1,2}\Aput{$$} \ncline{->}{2,2}{1,2}\Aput{$$}
        \ncline{->}{3,1}{2,1}\Aput{$$} \ncline{->}{3,2}{2,2}\Aput{$$}
        \ncline{->}{2,1}{2,2}\Aput{$$}
      \end{psmatrix}}
    \caption{\small {\em (a)} our initial problem setting,
    {\em (b)} assumptions about the relationships between
    different regimes of data collection are added by the
    inclusion of {\em intervention indicators} in the graph,
    as discussed in Section~\ref{ADAGs},
    {\em (c)} the effects of $A$ and $B$ on $Y$
    are jointly, but not individually,
    unconfounded.
    \label{Figure 1}}
\end{center}
\end{figure}

\newpar If we can perform an experiment, setting $A$ and $B$ to specific
values and observing the corresponding $Y$ outcomes (but not observing
$V$), the resulting data may help us predict the effect upon $Y$ of
intervening on $A$ and/or $B$.  But we can probe more deeply.  We can
investigate {\em context-specific} causal effects --- the effects of
$A$ and $B$ upon $Y$ in a context determined by some given value $v$
for $V$.  For example, if $A$ and $B$ are logical variables, then for
any fixed value $v$ of $V$ the $f$ function of Equation~(\ref{eq:f})
will take one of sixteen possible Boolean patterns, such as, for
example, $Y \ugualino A \vee B$, or $Y \ugualino \overline{A} \wedge
B$, and so on. Under appropriate assumptions, the researcher may be
able to infer that a certain pattern occurs in a random individual
with positive probability.  If the pattern is, say, $Y \ugualino A
\wedge B$ --- a pattern where the two effects are interdependent ---
one might take this as evidence that, in certain circumstances, $A$
and $B$ operate in the same mechanism.  \cite{Rothman1976},
\cite{Rothman1998}, \cite{Vanderweele2008}, \cite{Vanderweele2009},
\cite{VanderweeleRobinsAnnals} and \cite{Skrondal2003} have explored
this territory, and proposed a series of empirical conditions for
``interdependence'' of binary variables focused on mechanistic
interaction.  \cite{Vanderweele2010} extends this theory to
multi-level ordered categorical factors.

\newpar The mathematical form of the tests proposed here is similar to
those that the above authors have proposed for discrete causal
factors.  However, by introducing novel assumptions, we derive tests
valid in the more general case of categorical and continuous causal
factors, in any combination.

We also provide a different justification and different assumptions
for inference about mechanism, in a framework built around the above
notion of deep determinism.

Section~\ref{Coaction} introduces the concept of {\em interference} to
capture the idea of two variables, $A$ and $B$, influencing $Y$ by
operating through the same mechanism; this concept allows for
asymmetry in the way $A$ and $B$ interact.  Thus we say that $B$
interferes with $A$ in producing the event $Y\ugualino 1$ when $A$ and
$B$ are both causal factors for $Y$, and there exists a possible
intervention on $B$ which has the power of preventing any intervention
on $A$ from causing the event $Y\ugualino 1$.  This can occur without
also having $A$ interfering with $B$. We talk of {\em weak coaction}
[resp., strong coaction] when at least one [resp., either] of $A$ and
$B$ interferes with the other.

The above concepts are defined in terms of the behaviour of the system
under a (real or hypothetical) {\em interventional} regime, where $A$
and $B$ are {forced} to take on specific value values.  However in
Section~\ref{Testing for coaction} we show that the proposed tests can
be applied to data collected under under other regimes, {\em e.g.\/}
observational.  In Section~\ref{The core conditions}, the conditions
under which these tests are meaningful are studied in terms of
conditional independence properties of an Augmented Directed Acyclic
Graph (ADAG) representation of the problem (\cite{dawidisi}). The ADAG
will simultaneously represent the consensus causal theory about the
system under study, and assumptions about the behaviour of the system
across different regimes of data collection.  ADAGs are briefly
reviewed in Section~\ref{ADAGs}.  The scientific relevance of the
method and its practicality in complex study designs are illustrated
with the aid of two studies of the molecular determinants of coronary
artery disease, one of numerous areas in biomedical research where an
assumption of deep determinism could be defensible.

\section{\textcolor[rgb]{0.00,0.00,0.63}{\bf Interference and
    coaction}}
\label{Coaction}

Henceforth we make the deep determinism assumption of
Equation~\eqref{eq:f}.  The set of possible values of $A$ [resp., $B$,
$V$] is denoted by ${\cal A}$ [resp., ${\cal B}, {\cal V}$].

\begin{Definition} {\bf (Irrelevance)} Factor $B$ is {\em (causally)
    irrelevant} to $Y$ in context $V \ugualino v$, given $A$, if
  $f(a,b,v) = f(a,b',v)$ for all $a \in {\cal A}, b, b' \in {\cal B}$.
\end{Definition}

\begin{Definition}
  {\bf (Interference)} We say that $A$ {\em interferes} with $B$ in
  producing the event $Y\ugualino 1$ if, in some context $V\ugualino
  v$, $B$ is not irrelevant to $Y$ given $A$ and, for some $\hat a \in
  {\cal A}$ and all $b\in{\cal B}$,
  \begin{align}
    f(\hat a, {b},v) = 0.
    \label{asimmetrica}
  \end{align}
  \label{interference}
\end{Definition}
\noindent That is, in that context, there exists a value $\hat{a}$
such that, when we set $A\ugualino \hat{a}$, the event $Y\ugualino 1$
will never happen, whatever value we impose on $B$.

\noindent \begin{Definition} {\bf (Weak coaction)} We say that $A$ and
  $B$ {\em weakly coact\/} to produce the event $Y\ugualino 1$ if at
  least one of $A$ and $B$ interferes with the other to produce the
  event $Y\ugualino 1$.
  \label{weak coaction}
\end{Definition}

\noindent \begin{Definition} {\bf (Strong coaction)} We say that $A$
  and $B$ {\em strongly coact\/} to produce the event $Y=1$ if each of
  $A$ and $B$ interferes with the other to produce the event $Y=1$.
    \label{strong coaction}
  \end{Definition}

  \newpar {\bf Example (Logical):} Under a regime of intervention on
  variables $A \in \{0,1,2\}$ and $B \in \{0,1\}$, let the binary
  response $Y$ depend on these two variables according to the logical
  law $Y = (A\ugualino 2) \vee ((A \ugualino 1) \wedge (B\ugualino
  1))$.  Neither of $A$ or $B$ is irrelevant to $Y$.  Setting $A$ to
  the value 0 will prevent the event $Y\ugualino 1$, whatever the
  value we impose on $B$.  However, when $A$ is set to value 2, event
  $Y\ugualino 1$ will happen whatever value we impose on $B$.  Hence
  $B$ does {\em not} interfere with $A$, while
  $A$ interferes with $B$, in producing the event $Y\ugualino 1$. Thus,
  $A$ and $B$ coact weakly (but not strongly) in producing the event
  $Y\ugualino 1$.

\begin{figure}
  \begin{center}
    \begin{pspicture}(11,8) \psset{dipolestyle=elektor}
      \pnode(-1,2){Input} \pnode(12,2){Output} \pnode(4,0){C}
      \pnode(0,2){A} \pnode(4,0){C} \pnode(4,2){B} \pnode(4,4){D}
      \pnode(8,3){E} \pnode(8,0){F} \pnode(8,4){G} \pnode(8,2){H}
      \pnode(11,2){I} \switch(A)(B){$A_2$} \switch(B)(H){$A_1$}
      \rput(Input){${\cal G}$} \rput(Output){${\cal Y}$} \wire(B)(D)
      \wire(G)(H)
      \switch(H)(I){$V$} \switch(D)(G){$B$}
    \end{pspicture}
  \end{center}
  \caption{\small Electrical circuit illustration of coaction
    asymmetry. Imagine that an electrical voltage is applied between
    pins ${\cal G}$ and ${\cal Y}$. Let $Y=1$ indicate absence of
    current between these two pins. Let $Y=0$ indicate presence of
    current between these two pins. See main text for discussion of
    this example. \label{Circuit}}
\end{figure}
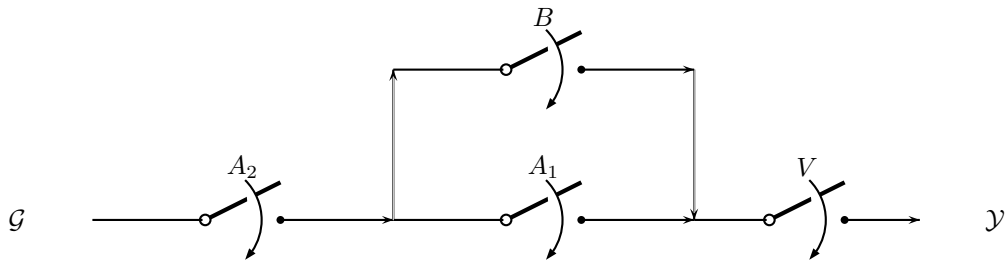

  \newpar {\bf Example (Electrical):} Consider the circuit of
  Figure~\ref{Circuit}, where we imagine an electrical voltage applied
  between pins ${\cal G}$ and ${\cal Y}$, and we take $Y=1$ [resp.,
  $Y=0$] to indicate that current flows [resp., does not flow] between
  these two pins.  Let the context variable be $U$, describing the
  unobserved state of the $U$--switch, each of the two possible states
  (OPEN, CLOSED) having positive probability.  Let variable $A$ index
  the four possible configurations of the $A$--switches, and variable
  $B$ the position of the $B$--switch.  The flow of current depends on
  the configuration of the switches via the well known deterministic
  laws of electrical circuits: this model thus satisfies deep
  determinism.  Then in context $U\ugualino$ CLOSED, variable $B$ is
  not irrelevant to $Y$ since, when $A_1$ is open and $A_2$ closed,
  acting on $B$ will have an effect on current flow. However, when
  $A_2 \ugualino$ is open, no intervention on the $B$--switch can restore
  the current flow. Hence, in context $U\ugualino $ CLOSED, variable
  $A$ interferes with $B$ in producing current flow.

  \newpar {\bf Example (Binary):} If $A$ and $B$ are binary,
  Equation~\eqref{eq:f} implies that, for a given value $v$ of $V$,
  the function $f$ takes one of sixteen possible patterns. First
  consider patterns $Y=TRUE$, $Y=FALSE$, $Y\ugualino A$, $Y\ugualino
  \overline{A}$, $Y\ugualino B$ and $Y\ugualino \overline{B}$. In all
  these patterns, at least one of $A$ or $B$ is irrelevant to $Y$, and
  therefore, by definition, neither of $A$ and $B$ interferes with the
  other in producing the event $Y\ugualino 1$. Next consider patterns
  $Y\ugualino A \vee B$, $Y\ugualino \overline{A} \vee B$, $Y\ugualino
  A \vee \overline{B}$ and $Y\ugualino \overline{A} \vee
  \overline{B}$, where the disjunctive form implies neither factor
  interferes with the other. Finally consider patterns $Y\ugualino A
  \wedge B$, $Y\ugualino \overline{A} \wedge B$, $Y\ugualino A \wedge
  \overline{B}$, $Y\ugualino \overline{A} \wedge \overline{B}$,
  $Y\ugualino (A \ugualino B)$ and $Y\ugualino (A \ne B)$, where
  neither of $A$ and $B$ is irrelevant, and where no value of $A$
  [resp., of $B$] produces the event $Y\ugualino 1$ unless $B$ [resp.,
  $A$] takes on a particular value. Hence, in these last six patterns,
  each of $A$ and $B$ interferes with the other in producing the event
  $Y\ugualino 1$. We conclude that, in the special case where $A$ and
  $B$ are binary, there can be no interference asymmetry between $A$
  and $B$: either they do or they do not interfere each with each
  other. Thus in this case weak and strong coaction coincide, and are
  essentially equivalent to the notion of interdependence given by
  \cite{Vanderweele2008}.

  \newpar {\bf Example (Biological determinism):} Suppose a genetic
  mutation $A$ can induce a structural change in protein $\alpha$,
  causing disease $Y$ in certain individuals when the protein is
  expressed normally. Hence $A$ is not irrelevant to $Y$.  Mutation
  $B$, located in the promoter region of the coding gene of $\alpha$,
  reduces the level of expression of $\alpha$.  As a consequence, in
  the above individuals, presence of $B$ prevents any structural
  disfunctionality in protein $\alpha$ from causing the disease.  In
  this case $B$ interferes with $A$ in causing disease $Y$ --- an
  example of what geneticists call ``epistasis''.

  \newpar We conclude this section with a remark. We have discussed
  ``coaction to produce''.  We could similarly have defined ``coaction
  to prevent''.  Coaction to prevent does not imply coaction to
  produce, nor {\em vice versa\/}. The scientific application and
  question of interest will usually dictate interest in one of the two
  directions.

  \section{\textcolor[rgb]{0.00,0.00,0.63} {\bf Monotonicity}}
  \label{Monotonicity}

  Sometimes we may be able to make assumptions about the ordering of
  the values of $Y$ in response to configurations of $A$ and $B$.  In
  the electrical example of the previous section, for example,
  increasing the number of switches in CLOSED position can never cause
  the current flow to be switched off.  Sometimes assumptions of this
  kind can be formulated as properties of {\em monotonicity}, as
  follows.

\begin{Definition}
  The effect of $A$ upon $Y$ is said to be {\em non-decreasing (with
    respect to $B$\/)} if, for any configuration $(b,v)$ of $(B,V)$,
  the following implication holds: $f({a},{b},v) = 1 \;\; \mbox{AND}
  \;\; {a}^{'} \ge {a} \,\Rightarrow\, f({a}^{'},{b},v) = 1.$
  \label{positive monotonic}
\end{Definition}
\begin{Definition}
  The effect of $A$ upon $Y$ is said to be {\em non-increasing (with
    respect to $B$)\/} if, for any configuration $(b,v)$ of $(B,V)$,
  the following implication holds: $f({a},{b},v) = 0 \;\; \mbox{AND}
  \;\; {a}^{'} \ge {a} \,\Rightarrow\, f({a}^{'},{b},v) = 0.$
  \label{negative monotonic}
\end{Definition}
\begin{Definition}
  The effect of $A$ upon $Y$ is said to be {\em monotonic (with
    respect to $B$)\/} if it is either non-decreasing or
  non-increasing with respect to $B$.
  \label{monotonic}
\end{Definition}

\begin{Definition}
  The effect of $A$ upon $Y$ is said to be {\em consistent (with
    respect to $B$)\/} if whenever, for any $(a_1, a_2)$ pair, the
  inequality $f(a_1, b, v) \ge f(a_2, b, v)$ holds for some $(b,v)$
  configuration, it holds for all $(b,v)$ configurations.
  \label{consistent}
\end{Definition}

\newpar Clearly monotonicity implies consistency; conversely, under
consistency we can re-order the values to yield monotonicity.
\cite{Berrington2007} discuss the situation where a change in the
value of $A$ may give rise to a reversal of the effect of $B$ upon
outcome.  Such {\em qualitative interaction} violates consistency.
Some authors consider qualitative interaction to be interpretable in
terms of mechanism. A formal test, different from the standard
statistical test for departures from additivity, should be performed
to assess whether a qualitative interaction could be due to chance
variation. One such test has been proposed by \cite{Azzalini1984}. The
tests proposed in this paper, which also differ from standard
statistical interaction tests, establish conditions for an
interpretation of interaction in terms of mechanism without
necessarily requiring that the underlying interaction be qualitative.

\section{\textcolor[rgb]{0.00,0.00,0.63} {\bf Augmented Directed
    Acyclic Graphs}}
\label{Augmented Directed Acyclic Graphs}

Coaction has been defined under a (real or hypothetical)
interventional regime.  The tests for coaction we shall later propose
may be applied more generally, such as when the data are
observational.  This, however, will require stringent assumptions, for
example that $V$ be conditionally independent of $A$ and $B$ and of
the way these two variables have been generated. In many applications
it will be possible, and is then helpful, to represent such
assumptions, in combination with further assumptions based on our
causal understanding of the problem, by means of an Augmented Directed
Acyclic Graph (ADAG).

\newpar Examples of ADAGs are given in Figure~\ref{Figure
  1}. Figure~\ref{Figure 1}{\em (b)} is an ADAG specialisation of the
simple problem setting of Figure~\ref{Figure 1}{\em (a)}.  An
important feature of ADAGs is inclusion of {\em intervention
  indicators\/}, exemplified in Figures~\ref{Figure 1}{\em (b)--(c)}
by nodes $\sigma_A$ and $\sigma_B$.  These nodes take values
indicating the particular regime, observational or experimental, under
which the values of a corresponding domain variable arise.  With $A$
and $B$ binary, for example, each of $\sigma_A$ and $\sigma_B$ will
have possible values in $(\emptyset, 0, 1)$, the interpretation being
that, when $\sigma_A = \emptyset$, the variable $A$ is generated
randomly by Nature, under the circumstances governing the
observational data; while $\sigma_A = {a} \in \{0,1\}$ indicates an
interventional setting in which value $a$ is imposed on $A$; and
similarly for $B$.  Although regime indicators are not random
variables, we can still query the ADAG, using the {\em
  $d$-separation\/} criterion of \cite{geiger90}, or the equivalent
{\em moralisation\/} criterion of \cite{lauritzen90}, to read off
conditional independencies implied by the graph.  These independencies
will generally reflect properties of the system under study {\em
  and\/} judgements about the way we expect the system to behave under
data collection regimes different from the actual one.  The graphs of
Figures~\ref{Figure 1}{\em (b)} and \ref{Figure 1}{\em (c)}, for example, embody
the conditional independence property, expressed in the notation of
\cite{Dawid1979}: $Y \,\indep (\sigma_A,\sigma_B) \mid (A,B)$, read as
``$Y$ is conditionally independent of $(\sigma_A,\sigma_B)$, given
$(A,B)$".  This represents an assumed property of invariance across
regimes: that once we know the values of $A$ and $B$, the distribution
of $Y$ will not further depend on the regime of data collection, as
represented by $(\sigma_A, \sigma_B)$. In other words, in these two examples, the
distribution of $Y$ does not depend on the way the $(A,B)$
configuration has arisen, be it observationally or interventionally.

\section{\textcolor[rgb]{0.00,0.00,0.63} {\bf The core conditions}}
\label{The core conditions}

Identifiability conditions for mechanistic interaction are typically
succinctly stated in terms of the effects of $A$ and $B$ on $Y$ having
to be ``unconfounded'', conditional on some observed variable $C$.  We
adopt a different approach, assuming a consensus ADAG representation
of the problem is available.  Conditions for validity of the test
proposed in the next section are then phrased in terms of conditional
independence properties of the ADAG. This discipline allows us to be
more precise in our claims than a formulation in terms of ``no
confounding''.  Another advantage of the ADAG-based approach is that
it makes it easier to relate the conditions for applicability of a
test to the substantive assumptions about the problem.

\newpar The assertion ``the (joint) effects of $A$ and $B$ on $Y$ are
unconfounded'' might be interpreted as saying that there exists an
observed variable $C$ such that the following two conditions are
satisfied:

\begin{align}
  \label{unconfoun}
  C \indep \sigma \hspace{1cm} \mbox{and} \hspace{1cm} Y \indep \sigma
  &\mid (A,B,C)
\end{align}
\noindent where $\sigma := (\sigma_A, \sigma_B)$, with possible values
$\sigma = (a,b)$, corresponding to setting $A=a, B=b$, and $\sigma =
(\emptyset,\emptyset)$, also denoted by $\sigma = \emptyset$, when
both $A$ and $B$ arise naturally.  In this case we say that $C$ is a
{\em sufficient covariate\/} for the joint effects of $A$ and $B$ on
$Y$ (\cite{hg/apd:aistats2010}).  In accordance with the ``back--door
criterion'' of~\cite{pearlbook}, under these conditions the joint
causal effect of $(A, B)$ on $Y$ will be estimable from observational
data when $C$ is also observed.  Note that these conditions need not
imply that $C$ is sufficient for the individual causal effects of each
of $A$ and $B$ on $Y$ (which would involve extending (\ref{unconfoun})
to apply also when only one factor is intervened on, {\em i.e.\/} for
$\sigma$ of the form $(a, \emptyset)$ or $(\emptyset, b)$).  Thus in
cases {\em (b)} and {\em (c)} of Figure~\ref{Figure 1}, $C \ugualino
\emptyset$ is sufficient for the joint effects of $A$ and $B$ on $Y$,
but is sufficient for the individual effects only in
Figure~\ref{Figure 1} {\em (b)}.

\newpar However, neither sufficiency for the joint effects nor sufficiency for
the individual effects is what we need to ensure applicability of the
test of the next section for a general regime of data collection.  In
our analysis, the additional observable variable $C$ must be a
function of the overall context variable $V$ featuring in the ``deep
determinism'' property (\ref{eq:f}).  Thus we can consider $V = (C,
U)$, with $C$ observed and $U$ unobserved.

We shall require the simultaneous validity of the following four {\em
  core conditions\/}:
\begin{Definition}
  \label{core conditions} {\bf (Core conditions)} There exists a
  (possibly empty) set $C$ of observable context variables and a set
  $U$ of (typically unobserved) context variables such that:
  \begin{enumerate}
  \item (deep determinism) $Y\ugualino f(A,B,C,U)$ for some
    deterministic function $f$, which is the same no matter how the
    variables $(A, B, C, U)$ are generated.
  \item $Y \indep \sigma \mid (A,B,C,U)$
  \item $U \indep (A,B,\sigma) \mid C $,
  \item $A \indep B \mid (C, \sigma)$.
  \end{enumerate}
\end{Definition}
\noindent
Whenever Condition~1 is satisfied, we say that $Y$ is {\em
  functional\/} with respect to $(A,B,C,U)$.  Condition~2 essentially
repeats the second part of Condition~1, but it is helpful to display
it explicitly.  Condition~3 says that, conditionally on $C$,
variable $U$ has
the same distribution in all regimes, and is independent of $A$ and
$B$ (this will hold, in particular, if the full context variable
$(C,U)$ has the same distribution in all regimes and is independent of
$A$ and $B$); while Condition~4 requires $A$ and $B$ to be
independent, given $C$, in the observational regime (this property
necessarily holding when $A$ and $B$ are set by intervention).

The following theorem can be proved straightforwardly using general
properties of conditional independence (\cite{Dawid1979},
\cite{pearlbook}).

\begin{theorem} Core conditions 2 and 3 imply $Y \indep \sigma \mid
  (A,B,C)$.
\end{theorem}

\newpar Our core conditions imply the second condition of
Equation~\eqref{unconfoun}, but not the first. It seems useful and
instructive to discuss the differences between the two sets of
conditions with the aid of examples.  In the following examples
interest focuses on testing coaction of variables $A$ and $B$ in
producing the event $Y\ugualino 1$, based on observational data about
variables $(A,B,Y)$ and, sometimes, a further variable $Z$.

\newpar Figures~\ref{Figure 1}{\em (b)--(c)} satisfy
the conditions of Equation~\eqref{unconfoun} when $C \ugualino
\emptyset$. In both these examples, the
distribution of $Y$ given $(A,B)$ does not depend on the way the
configuration of values of $(A,B)$ is generated, be it
observationally or by intervention. However, while Figure~\ref{Figure 1}{\em (b)}
satisfies the core conditions once we assume $Y$ to be functional
with respect to $(A,B,U)$), Figure~\ref{Figure 1}{\em (c)} violates
core condition 4.

\begin{figure}
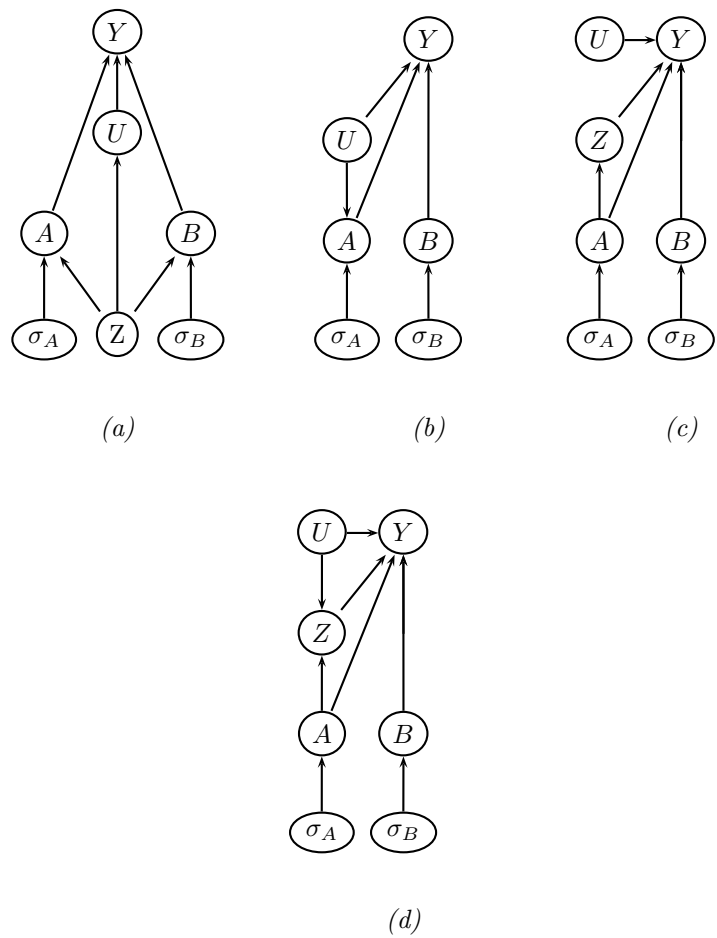

  \begin{center}
    \psframebox[fillstyle=solid,fillcolor=white,
    linestyle=none,framesep=.4,arrowsize=2]{%
      \begin{psmatrix}[ref=t,rowsep=0.7cm,colsep=0.2cm]
        &\Show{$Y$}\\
        &\Show{$U$}&&\\
        \Show{$A$}&&\Show{$B$}\\
        \Show{$\sigma_A$}&\Show{Z}&\Show{$\sigma_B$}\\
        &{\em (a)} \ncLine{->}{2,2}{1,2}\Aput{$$}
        \ncLine{->}{3,1}{1,2}\Aput{$$} \ncline{->}{3,3}{1,2}\Aput{$$}
        \ncline{->}{3,4}{1,2}\Aput{$$} \ncline{->}{4,2}{2,2}\Aput{$$}
        \ncline{->}{4,2}{3,1}\Aput{$$} \ncLine{->}{4,2}{3,3}\Bput{$$}
        \ncline{->}{4,1}{3,1}\Aput{$$} \ncLine{->}{4,3}{3,4}\Bput{$$}
        \ncLine{->}{4,4}{3,4}\Bput{$$} \ncLine{->}{4,3}{3,3}\Bput{$$}
      \end{psmatrix}} \hspace{-0.1cm}
    \psframebox[fillstyle=solid,fillcolor=white,
    linestyle=none,framesep=.4,arrowsize=2]{%
      \begin{psmatrix}[ref=t,rowsep=0.7cm,colsep=0.2cm]
        &\Show{$Y$}\\
        \Show{$U$}&&&\\
        \Show{$A$}&\Show{$B$}&\\
        \Show{$\sigma_A$}&\Show{$\sigma_B$}&\\
        &{\em (b)} \ncLine{->}{2,1}{1,2}\Aput{$$}
        \ncLine{->}{2,1}{3,1}\Aput{$$} \ncLine{->}{3,1}{1,2}\Aput{$$}
        \ncline{->}{3,2}{1,2}\Aput{$$} \ncLine{->}{4,2}{3,2}\Bput{$$}
        \ncline{->}{4,1}{3,1}\Aput{$$}
      \end{psmatrix}} \hspace{-0.1cm}
    \psframebox[fillstyle=solid,fillcolor=white,
    linestyle=none,framesep=.4,arrowsize=2]{%
      \begin{psmatrix}[ref=t,rowsep=0.7cm,colsep=0.2cm]
        \Show{$U$}&\Show{$Y$}\\
        \Show{$Z$}&\\
        \Show{$A$}&\Show{$B$}\\
        \Show{$\sigma_A$}&\Show{$\sigma_B$}\\
        &{\em (c)} \ncLine{->}{1,1}{1,2}\Aput{$$}
        \ncLine{->}{2,1}{1,2}\Aput{$$} \ncline{->}{2,2}{1,2}\Aput{$$}
        \ncLine{->}{3,2}{1,2}\Bput{$$} \ncline{->}{3,1}{1,2}\Aput{$$}
        \ncline{->}{3,1}{2,1}\Aput{$$} \ncline{->}{4,1}{3,1}\Aput{$$}
        \ncline{->}{4,2}{3,2}\Aput{$$}
      \end{psmatrix}} \hspace{-0.1cm}
    \psframebox[fillstyle=solid,fillcolor=white,
    linestyle=none,framesep=.4,arrowsize=2]{%
      \begin{psmatrix}[ref=t,rowsep=0.7cm,colsep=0.2cm]
        \Show{$U$}&\Show{$Y$}\\
        \Show{$Z$}&\\
        \Show{$A$}&\Show{$B$}\\
        \Show{$\sigma_A$}&\Show{$\sigma_B$}\\
        &{\em (d)} \ncLine{->}{1,1}{1,2}\Aput{$$}
        \ncLine{->}{1,1}{2,1}\Aput{$$} \ncLine{->}{2,1}{1,2}\Aput{$$}
        \ncline{->}{2,2}{1,2}\Aput{$$} \ncLine{->}{3,2}{1,2}\Bput{$$}
        \ncline{->}{3,1}{1,2}\Aput{$$} \ncline{->}{3,1}{2,1}\Aput{$$}
        \ncline{->}{4,1}{3,1}\Aput{$$} \ncline{->}{4,2}{3,2}\Aput{$$}
      \end{psmatrix}}
  \end{center}
  \caption{\small ADAG representations of problem examples discussed
    in the main text.  \label{ADAGs}}
\end{figure}

\newpar Now consider the example of Figure~\ref{ADAGs}{\em (c)}.  If
the researcher engaged in a test of coaction between $A$ and $B$
follows the ``no confounding'' conditions of
Equation~\eqref{unconfoun}, he or she will notice that these
conditions are satisfied for $C\ugualino \emptyset$, and might
therefore proceed to perform the test without conditioning on $Z$. By
contrast, if the researcher follows the core conditions of
Definition~\ref{core conditions}, he/she will notice that ignoring $Z$
(that is, setting $C\ugualino \emptyset$), is valid {\em only\/} under
the assumption that $Y$ is functional with respect to $(A,B,U)$. This
appears to be a tremendously stringent assumption, which we may accept
only if, for every value of $U$, variable $Y$ is a deterministic
function of $(A,B,Z)$ {\em and} $Z$ is a deterministic function of
$A$. A more appropriate choice, according to the conditions of
Definition~\ref{core conditions}, is to set $C \ugualino Z$.  The
latter choice would make more sense from a further point of view, that
is, it would test coaction of the effect of $B$ (on $Y$) and the {\em
  direct} effect of $A$ (on $Y$), unmediated by $Z$. In summary, in
this example, the two sets of conditions lead to different choices, in
the sense that the best choice according to the core conditions
violates the ``no confounding'' conditions of
Equation~\eqref{unconfoun}.

\newpar Many of the above considerations also apply to the example of
Figure~\ref{ADAGs}{\em (d)}. In particular, in this last example,
setting $C\ugualino \emptyset$ would appear a safe option according to
the `no confounding'' conditions of Equation~\eqref{unconfoun}. And it
would, in addition, satisfy core conditions {\em (2)} to {\em (4)}.  A
possible difficulty with this choice would however arise when
negotiating core condition 1. In the light of core condition 1, choice
$C \ugualino \emptyset$ means we are ready to assume $Y$ to be
deterministic when we condition on $(A,B,U)$, {\em but not} on
$Z$. This is sensible only if we believe $Z$ to be {\em itself} is a
deterministic function of its predecessors in the graph.  Neither does
the option $C\ugualino Z$, in this example, solve the problem.  For
conditioning on $Z$ will typically introduce dependence between $U$
and $A$, violating core condition 3.

\section{\textcolor[rgb]{0.00,0.00,0.63}{\bf Testing coaction}}
\label{Testing for coaction}

We now present a test for coaction of variables $A$ and $B$ in
producing the event $Y\ugualino 1$, assuming that there exists a
(possibly empty) set $C$ of observed variables such that the core
conditions of the previous section are valid.  We allow $A$ and $B$ to
be ordered categorical or continuous variables, in any combination.
If either variable is not binary, we consider some dichotimisation of
its range.  Thus for $A$ we would choose a threshold $\tau_A$ and
define $\alpha := \{a \in{\cal A}: a > \tau_A\}$, $\overline{\alpha}
:= \{a\in{\cal A}: a \le \tau_A$\}.  Similarly for $B$ we would have
$\tau_B, \beta, \overline{\beta}$.  We also use $\alpha$ to denote the
truth-value ($0$ or $1$) of the event $A \in \alpha$, {\em etc\/}.

In the sequel, all probabilities are computed under the observational
regime $\sigma = \emptyset$.

\noindent For $i, j = 0, 1$, let
\begin{align*}
  R_{ijc} &:= P(Y\ugualino 1 \mid \alpha \ugualino i, \beta \ugualino
  j, C\ugualino c)\\
  &\stackrel{\mbox{\footnotesize core condition 3}}{=} \; \int_u
  R_{ijc}(u)\; P(u \mid C\ugualino c)\; \d u
\end{align*}
where, for any value $u$ of $U$,
\begin{align}
  \label{iniziale}
  R_{ijc}(u) &:= P(Y\ugualino 1 \mid \alpha = i, \beta = j,
  C\ugualino c, U\ugualino u) \nonumber\\
  &= \int_{\alpha=i} \d a \int_{\beta=j} \d b\, P(Y\ugualino 1 \mid a,
  b, C\ugualino c, U\ugualino u)\; P(a,b \mid \alpha = i, \beta = j,
  C\ugualino c, U\ugualino u )
  \nonumber \\
  &\stackrel{\mbox{\footnotesize core conditions 1, 2}}{=}
  \int_{\alpha=i} \d a \int_{\beta=j} \d b\, f(a,b,c,u)\; P(a,b \mid
  \alpha = i, \beta = j, C\ugualino c, U\ugualino u)
  \nonumber \\
  &\stackrel{\mbox{\footnotesize core condition 3}}{=} \;
  \int_{\alpha=i} \d a \int_{\beta=j} \d b\, \frac{f(a,b,c,u)\; P(a,b
    \mid C\ugualino c)} {P(\alpha = i, \beta = j \mid C\ugualino c)}
  \nonumber \\
  &\stackrel{\mbox{\footnotesize core condition 4}}{=} \;
  \int_{\alpha=i} \d a\, {P(a \mid \alpha = i, C\ugualino c)}
  \int_{\beta=j} \d b\,f(a,b,c,u)\; P(b \mid \beta = j, C\ugualino c).
\end{align}
\begin{Definition}
  \label{insensatez}
  Variable $A$ is said to be {\em $\alpha$-insensitive} with respect
  to $Y$ if the following implication is valid
  for all $(b,c,u)$:
  \begin{align}
    \label{ordinale1}
    \mbox{IF} \;\;&f(a,b,c,u) = 0 \;&\mbox{for some} \; a \in
    \alpha \;\mbox{AND} \;\; a^{'} \ge a \;\; &\mbox{ THEN} &f(a^{'},b,c,u) &= 0
  \end{align}
\end{Definition}
\noindent We similarly define the $\beta$-insensitivity property for
$B$.  Trivially $\alpha$-insensitivity holds if $\alpha$ consists of a
single point. We are now ready to state the main theorem:
\begin{theorem}
  \label{teorema principale}
  Let the binary outcome variable $Y$ depend on observed variables
  $(A,B,C)$ and on unobserved variable $U$, where $A$ and $B$ are
  allowed to be ordered categorical or continuous, in any combination
  of these two types.  Let the effect of $A$ [resp., $B$] upon $Y$ be
  monotonic with respect to $B$ [resp., $A$], and suppose that, for
  some dichotomizations of $A$ and $B$, and some value $c$ of $C$:

  \begin{align}
    \label{condizioneprimoteorema}
    R_{11c}-R_{10c}-R_{01c} &> 0.
  \end{align}
  \noindent Then under the core conditions and the
  $\alpha$--insensitivity property for $A$, variable $B$ interferes
  with $A$ in producing the event $Y=1$.  Similarly, whenever the
  $\beta$--insensitivity property holds for $B$, variable $A$
  interferes with $B$ in producing the event $Y=1$; in either case $A$
  and $B$ weakly coact to produce the event $Y \ugualino 1$.
\end{theorem}

\begin{proof}
  \small Equation~\eqref{condizioneprimoteorema} can be expressed as
  \begin{align}
    \int \left[ R_{11c}(u) - R_{10c}(u) - R_{01c}(u) \right] p(\d u
    \mid C\ugualino c) &> 0.
    \label{prima}
  \end{align}
  It follows that there is a positive probability of obtaining a value
  $u^{*}$ of $U$ such that $R_{11c}(u^*) - R_{10c}(u^*) - R_{01c}(u^*)
  > 0$; in particular, $R_{11c}(u^*) - R_{10c}(u^*) > 0$.  Thus, using
  \eqref{iniziale},
  \begin{align*}
    \int_{a \in \alpha} \d a\, P(a \mid A \in \alpha, C\ugualino c)
    &\left[ \int_{b \in \beta} \d b\, f(a,b,c,u^*)\, P(b \mid B \in
      \beta, C\ugualino c) \right. \\
      &\left. - \int_{b \in \overline{\beta}} \d b
      \,f(a,b,c,u^*)\; P(b \mid B \in \overline{\beta}, C\ugualino c )
    \right] &> 0,
  \end{align*}
  \noindent from which it follows that there exists a value $a_1 \in
  \alpha$ such that
  \begin{align}
    \int_{b \in \beta} \d b \,f(a_1,b,c,u^*)\; P(b \mid B \in \beta,
    C\ugualino c) > \int_{b \in \overline{\beta}} \d b
    \,f(a_1,b,c,u^*)\; P(b \mid B \in \overline{\beta}, C\ugualino c).
    \label{terza}
  \end{align}

  \noindent Since the left-hand-side of the above inequality is thus
  positive, and $f = 0$ or $1$, we must have
\begin{align}
    \label{first bit}
    f(a_1,b_1,u^{*},c) &=1 &\mbox{for some}\;\; b_1 \in \beta.
  \end{align}
  Also we cannot have have $f(a_1,b,u^{*},c) = 1$ for all
  $b\in\overline{\beta}$, since in this case the right-hand side of
  \eqref{terza} would equal 1, whereas the left-hand side can not
  exceed 1.  We deduce that
  \begin{align}
    \label{second bit}
    f(a_1, b_2, u^{*},c) &= 0, &\mbox{for some} \;\; b_2 \in
    \overline{\beta}.
  \end{align}
  \noindent Because Equation~\eqref{prima} is symmetrical in $A$ and
  $B$, we similarly obtain:
  \begin{align}
    \label{scambiata1}
    f(a_2,b_3,u^{*},c) &=1, &\mbox{for some} \;\; a_2 \in \alpha \;\; \mbox{and}  \;\;  b_3 \in \beta,\\
    \label{scambiata2}
    f(a_3,b_3,u^{*},c) &=0, &\mbox{for some} \;\; a_3 \in
    \overline{\alpha}.
  \end{align}

  \newpar Under the assumed monotonicity of the effect of $B$ upon
  $Y$, and remembering that $\beta$ lies above $\overline{\beta}$,
  Equations~\eqref{first bit}--\eqref{second bit} imply that $f$ is
  non-decreasing with $B$ for any configuration of
  $(A,C,U)$. Similarly,
  Equations~\eqref{scambiata1}--\eqref{scambiata2} imply that $f$ is
  non-decreasing with $A$ for any configuration of $(A,C,U)$.
  Equations~\eqref{first bit}~\eqref{second
    bit}~\eqref{scambiata1} and~\eqref{scambiata2} tell us that there
  is a context $(U, C) = (u^{*},c)$ where variables $A$ and $B$ are
  {\em not} irrelevant to $Y$ with respect to each other.  Then
  according to Definition~\ref{interference}, in order to prove that
  $B$ interferes with $A$ in producing the event $Y \ugualino 1$, we
  only need prove that, for some value imposed on $B$, no value of $A$
  will produce the event $Y \ugualino 1$, that is:
  \begin{align}
  f(a, b_2, u^{*},c) \; &= \; 0 \;\; \;\;\; \forall a .
  \label{finale}
  \end{align}
  \newpar In fact, the following two
  implications follow from Equation~\eqref{second bit}:
  \begin{align*}
  a^{*} < a_1 \; &\Rightarrow \; f(a^{*}, b_2, u^{*},c) &\stackrel{\mbox{\footnotesize $f$
  non-decreasing with $A$}}{=} \; 0\\
  a^{*} \ge a_1 \;\;\mbox{AND} \;\; a_1 \in \alpha \; &\Rightarrow \;
  f(a^{*}, b_2, u^{*},c) &\stackrel{\mbox{\footnotesize $A$ is $\alpha$-insensitive wrt $Y$}}{=} \; 0
  \end{align*}
  \noindent from which Equation~\eqref{finale} follows.
  We then conclude that, under an assumed {\em $\alpha$--insensitivity
    condition} for $A$, Equation~\eqref{condizioneprimoteorema}
  implies that variable $B$ interferes with $A$ in producing the event
  $Y=1$. Similarly we can prove that, under an assumed
  $\beta$--insensitivity property for $B$, variable $A$ interferes
  with $B$ in producing the event $Y=1$, which completes the proof.
\end{proof}

\newpar {\bf Remark 1:} The theorem
holds also in those situations where we
can have its conditions satisfied by
an appropriate recoding of $A$ and $B$.

\newpar {\bf Remark 2:} The theorem can be applied
conditional on the generic individual belonging
to a particular population stratum defined
on the basis of $(A,B)$.

\newpar  The following example illustrates
the two remarks above.  Consider a discrete variable $A \in
\{1,2,3,4\}$ and a continuous variable $B$ on the $(0,1)$
interval.  Assume monotonicity of the effects
of $A$ and $B$ on $Y$, and let us restrict attention
to the stratum of individuals with $A > 1$. Let us then
recode variable $A$ by setting
$A^{*} :=  5-A$. Then suppose
for $\alpha: \{A^{*} \ugualino 3 \}$ and $\beta: \{B > 0.5 \}$ that the
data strongly support the inequality $R_{11}-R_{01}-R_{10} > 0$.
Because $\alpha$ consists of a single point, and therefore
$A$ is $\alpha$-insensitive with respect to $Y$,
we may conclude that $B$ interferes with $A$
in producing the event $Y \ugualino 1$. The reverse inference, that
$A$ interferes with $B$, is possible if $B$ is $\beta$-insensitive
with respect to $Y$, but this assumption may be problematic since
$\beta$ does not consist of a single point.

\section{\textcolor[rgb]{0.00,0.00,0.63}{\bf Examples}}
\label{Examples}

We discuss the examples of Figures~\ref{ADAGs}{\em (a)---(b)}.

\noindent {\bf Example of Figure~\ref{ADAGs}{\em (a)}}
Let $Y$ be an indicator of disease, depending on a pair $(A,B)$ of
genetic variants in linkage equilibrium with each other; and let
covariate $Z$, representing genealogical information, say,
be sufficient for the effects
of $A$ and $B$ on $Y$. Then the graph of Figure~\ref{ADAGs}{\em (a)} might be an
acceptable representation of the problem.  Suppose further there is
consensus that $Y$ is functional with respect to $(A,B,U)$, for
example because the effects of the two variants on $Y$ are thought to
operate through a common molecular mechanism.  Then the core
conditions are satisfied if we take $C \equiv Z$, and so observational
$(A,B,Z,Y)$ data can be used to test for $A$--$B$ coaction to produce
the event $Y\ugualino 1$.

\noindent {\bf Example of Figure~\ref{ADAGs}{\em (b)}} In this
example, where $C$ is necessarily empty, node $U$ is {\em not}
independent of $A$, which
violates core condition 3.
  Consequently a set of observational
$(A,B,Y)$ data will typically not suffice for us to be able to test
productive coaction of $A$ and $B$ by using the proposed method.

\section{\textcolor[rgb]{0.00,0.00,0.63}{\bf Relations with previous
    work}}
\label{Relations with previous work}

In certain formal frameworks for ``statistical causality'', including
Pearl's structural equation formulation (\cite{pearlbook}, chapter 7)
and the potential response framework of \cite{Rubin2005}, it is
possible to construct a totally fictitious mathematical variable $V$
which makes (\ref{eq:f}) true by mathematical fiat.  Our approach
differs in that we conceive of the context variable $V$ as both real
and relevant --- and thus in principle observable; its relationships
with the remaining variables in the problem need be negotiated and
explicitly represented in the causal model.  This has practical
consequences for data analysis.  Consider, for example,
Figure~\ref{ADAGs}{\em (c)} and {\em (d)}.  These two examples differ
only in that, in the former, on the basis of contextual knowledge, we
judge the unobserved ``context variables'' $U$ that differentiate the
possible behaviours of $Y$ in response to $(A,B)$ to be {\em a
  priori\/} unrelated to $Z$, whereas in the latter example, these
unknown variables are judged to act as unobserved confounders of the
effect of $Z$ upon $Y$.  We have seen that this difference has
consequences on our decision to apply the method, and whether or not
we should condition on $Z$.

\newpar Also, our method replaces the generic assumption of ``the
effect of $A$ and $B$ on $Y$ is not confounded given $C$'' with a
formal set of independencies (the core conditions) that need to be
satisfied by the causal model.  We have seen in Section~\ref{The core
  conditions} that this formal method can capture important
differences between different applications.

\section{\textcolor[rgb]{0.00,0.00,0.63}{\bf Illustrative study:
    rs1333040 coacts with statins}}
\label{Illustrative study: rs1333040 coacts with statins}

Within the Italian genetic study of early-onset myocardial infarction
(\cite{Ardissino2010}), between 1996 and 2002, an incident sample of
2050 cases was selected on the basis of an hospitalization for
myocardial infarction (MI) between age 40 and age 45, over a set of
125 Coronary Care Units spread nationwide. After entering the study,
each sample subject produced a blood sample from which plasma was
separated and DNA extracted, and was then prospectively monitored for
an average of 12 years of follow-up. Let the outcome of the follow-up
be represented by a binary variable, $Y$, indicating whether a
re-infarction or cardiovascular death were observed ($Y\ugualino 1$),
or not observed ($Y\ugualino 0$) within a period of 120 months from
study entry.

\newpar The research group agrees on the assumptions represented in
the ADAG of Figure~\ref{Figura studio}.  According to the graph, each
case is characterized by the following variables. Variable $G$
is a function of the genotype at rs1333040, a single nucleotide
polymorphism (SNP) located in chromosomal region 9p21.3. We
define $G$ to take value $1$ in presence of two copies of
the major rs1333040 allele, and value $0$ otherwise.
Variable $Z$ is the severity of coronaropathy at study entry.
Variable $T$ is the calendar year at study entry.
Variable $U$ represents a set of unknown confounders.
Variable $S$ indicates whether the subject was assigned to statin
treatment right after study entry ($S\ugualino 1$) or never after
study entry ($S\ugualino 0$), and
$I$ indicates presence/absence of hypercholesterolemia at study
entry. Variable $T$ here acts as a surrogate for relevant factors that vary
with calendar time. These include therapy evolution, progress of
medical knowledge and impact of legislation. These factors are assumed
to influence both medical practice, specifically concerning use of
statins, and the clinical outcome $Y$. During the study period,
National Guidelines concerning use of statins had not yet come into
force, and the decision whether or not to administer statins to
patients of the kind we are studying was taken more or less
randomly by the recruiting Coronary Care Unit, though to
some extent dependent on whether or not the patient was found
to have hypercholesterolemia at study entry. This
is accounted for in the graph by the $I \rightarrow S$ arrow. The
graph also conservatively allows that susceptibility to
hypercholesterolemia may depend on the genotype at the SNP of
interest, although evidence in support of this has never been found.

\begin{figure}
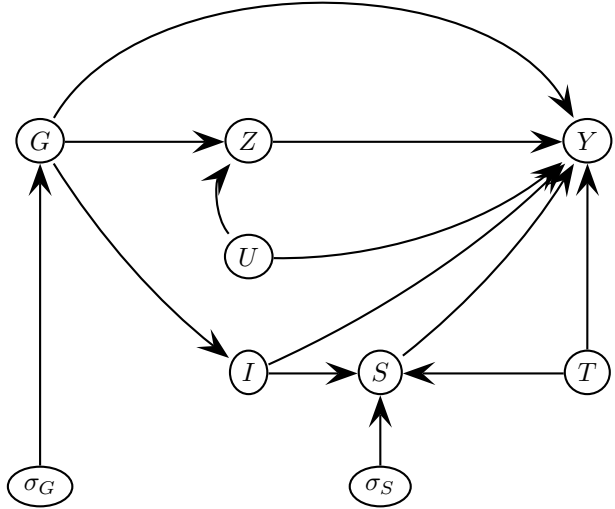

  \begin{center}
    \psframebox[fillstyle=none,fillcolor=pink,
    linestyle=none,framesep=.4]{%
      \begin{psmatrix}[ref=t,rowsep=0.9cm,colsep=1cm]
        \rule{0cm}{3cm}
        &&&&&&\\
        &\Show{$G$}&&\Show{$Z$}&&&\Show{$Y$}\\
        &&&\Show{$U$}&&&\\
        &&&\Show{$I$}&\Show{$S$}&&\Show{$T$}\\
        &\Show{$\sigma_G$}&&&\Show{$\sigma_S$}&& \psset{arrows=->,arrowsize=4pt
          6,fillcolor=white,fillstyle=solid}
        \ncline{5,2}{2,2}
        \ncline{2,2}{2,4}
        \ncline{4,4}{4,5}
        \ncline{4,7}{4,5}
        \ncline{4,7}{2,7}
        \ncline{4,6}{4,5}
        \ncline{2,4}{2,7}
        \ncline{5,5}{4,5}
        \psset{arrows=->,arrowsize=4pt
          6,fillstyle=none}
          \ncarc[arcangle=-10]{2,2}{4,4}
        \ncarc[arcangle=-10]{4,5}{2,7}
        \ncarc[arcangle=60]{2,2}{2,7}
        \ncarc[arcangle=-20]{3,4}{2,7}
        \ncarc[arcangle=-10]{4,4}{2,7}
        \ncarc[arcangle=40]{3,4}{2,4}
      \end{psmatrix}}
  \end{center}
  \caption{\small ADAG representation of our illustrative study of
  coaction between a gene tagged by single nucleotide polymorphism
  rs1333040 and statin treatment in producing myocardial
  infarction. See main text for a justification
  of the causal relationships depicted in this graph. \label{Figura studio}}

\end{figure}

\newpar Instead of performing separate analyses within strata of
$(T,I)$, we restrict analysis to the stratum of patients with
hypercholesterolemia $(I \ugualino 1)$, and
assume that, in this stratum, the effect of $T$ does not
interact with $G$ and $S$.  We then model the effect
of $(G,S,T)$ on $Y$ in the stratum of patients with $I=1$
via the following linear
risk Bernoulli model:
\begin{equation}
  \left\{
    \begin{array}{lll}
      Y&\sim& {\rm Bernoulli}(\pi),\\
      \pi &=& \alpha + \phi_{S=0} +\phi_{G=1} + \gamma_{\footnotesize (S=0) \times (G=1)} +
      \delta_{t}\; T,
    \end{array}
  \right.
  \label{modello regressione}
\end{equation}
\noindent where $\delta_{t}$ represents a linear effect of
calendar year, in years since 1970. If our data provide
evidence of a departure of
parameter $\gamma_{(S=0) \times (G=1)}$
from zero, we say that variables $S$ and $G$
interact {\em statistically} in the stratum
of hypercholesterolemic patients and, indeed,
the results
shown in Table~\ref{Table
  1} support this conclusion.
  The data seem to tell us that the
  beneficial effect of statins, in terms of
  reduction of risk of re-infarction in a hypercholesterolemic
  patient,
  is stronger in
  patients with $G = 0$.
  And that the highest risk is found in those
hypercholesterolemic patients with $G = 0$ who do
not receive statins.

\newpar Let's now shift from predictive
to mechanistic inference,
by examining whether the
interaction between variables $S$ and $G$
can be upgraded from ``statistical'' to ``mechanistic''.
In order to do this, we need to use
a different statistical test,
and to be explicit about the set
of (fairly strong) assumptions discussed in the previous section.
One of these is monotonicity of the
effects, which appears to be reasonable, since it does not
require prior knowledge of
the ``deleterious'' allele of the SNP.
Next, we need to assume that the core conditions hold.
Define $C = (T,I)$. With this choice, core
conditions 2 to 4 are satisfied, although core condition 1 --- that
$Y$ be a deterministic function of $(G,S,T,I,U)$ --- could be
problematic here unless we assume that, for any given value of $U$,
variable $G$ influences $Z$ and $Y$
through the same molecular mechanism whereby interference with the
effect of statin takes place. After accepting the core conditions, in
accordance with the theory of Section~\ref{Testing for coaction},
we partition the possible values of the rs1333040 genotype
into the set
$\alpha$ and its complement $\overline{\alpha}$.
We define $\alpha$ to indicate presence of
two copies of the most frequent rs1333040 allele,
corresponding to $G=1$,
so that $\overline{\alpha}$ will represent the remaining
two genotypic categories.
We define $\beta$ to
indicate that the patient is given statins,
corresponding to $S = 1$,
and we define $\overline{\beta}$
to indicate that the patient is {\em not} given
statins, corresponding to $S = 0$.
Since each of $\alpha$ and $\beta$ contains just one
value of the corresponding variable, $\alpha$-insensitivity
and $\beta$-insensitivity hold
in this case.

\newpar It is easy to show that, given $T=t$, the above
model implies $R_{11t}-R_{10t} -R_{01t} =  \gamma_{\footnotesize
(S=0) \times (G=1)} - \alpha - \delta_{t}t$.
This quantity, according to Table~\ref{Table 1}, is
significantly greater than zero for all relevant
values of $T$. Hence, in the light of our theory
and under the assumptions discussed above, we conclude that
$G$ and $S$ strongly coact to produce re-infarction.
The interpretation may be phrased in a number of ways.
One is to say that there exists some context
in which hypercholesterolemic patients with the
$G=1$ genotype are safe from re-infarction, whether or not
they take statins, whereas those with $G=0$ develop
or avoid re-infarction depending on assumption of statins.
A counterfactual rephrasing of this is to say that
some patients with $G=0$, who developed
re-infarction, would not have developed it,
had they received statins. All this can be interpreted to
suggest that statins and some gene tagged by rs1333040
influence susceptibility to reinfarction through a common
pathway, which motivates a future effort to
identify which gene is this, and what is its function.
Some researchers might have got to the same conclusions
from the results of the regression analysis, without
consideration of the theoretical framework proposed in
this paper. In our opinion, that would be careless.
Not only do such conclusions require
a statistical test of the kind proposed in this paper,
which differs from a standard interaction test, but also,
they require explicit consideration
of the (fairly strong) assumptions we have discussed in this paper.

\section{\textcolor[rgb]{0.00,0.00,0.63}{\bf Illustrative study:
    rs4620585 coacts with smoking}}
\label{Illustrative study 2}

Each of the cases in the study of the previous section was paired with
a control, matched by age and geographical region of origin.  After
eliminating individuals with missing data, 1666 controls remained
available for the analysis.  In this section we concentrate on
``smoking habit'', a binary indicator obtained by dichotomizing an
(imprecisely recorded) daily number of cigarettes.
We tested for possible coaction
between smoking habit and one or more SNPs of a list of ten
candidates from an independent study, an interesting signal
being found at SNP rs4620585 of human chromosome 1, never
previously been associated with a disease. The remaining
discussion restricts attention to SNP rs4620585.
Let $A$ signify rare rs4620585 homozygosity (RRH), and $B$ signify
``smoker''. Let $Y$ represent occurrence of early MI.
We assume that
the core conditions, and in particular
condition 4, hold in this problem, once we assume
(in accord with current knowledge) that the gene implicated
by rs4620585 has no influence on smoking habit or addiction to
nicotine.

\newpar On the basis of our data, we performed
a linear-odds regression of the case-control indicator
on SNP rs4620585 and smoking
habit. This analysis yielded the estimated coefficients of
Table~\ref{Table 3}. Because our ``early MI'' endpoint
is rare, we may safely assume
that the selection effect implicit in the
case-control study affects
the interaction parameter $\gamma$ and
the intercept  $\alpha$, in principle
estimable only through a prospective study, only
through multiplication by a common, unknown,
positive constant. Hence
we may take positivity
of $(\gamma-\alpha)$ to imply positivity of the linear combination
$R_{11}-R_{01}-R_{10}$ of prospective risks.
Since Table~\ref{Table 3} shows the quantity $R_{11}-R_{01}-R_{10}$ to be
significantly greater than zero (no multiple testing adjustment),
we conclude in favour of a potential
mechanistic interaction between SNP rs4620585 and smoking.
One interpretation of this result is to say that there
are circumstances in which some patients, by virtue of a
beneficial variant tagged by rs4620585, are safe from an early
MI regardless of their smoking, whereas patients without that
variant, who in the same circumstances developed an early MI,
would have avoided it, had they {\em not}
smoked.

\section{\textcolor[rgb]{0.00,0.00,0.63}{\bf Discussion}}
\label{Discussion}

Statistical interaction --- departure from some parametric model of
independent effects of explanatory variables --- is not necessarily
interpretable as reflecting an underlying mechanism, not least because
most statistical models are mathematical fictions
(\cite{Clayton2009}).  This is especially true when the modeller has
to negotiate {continuous} explanatory variables.  Our proposed
sufficient conditions for declaring coaction between continuous
variables do not invoke specific parametric forms of dependence, and
appear to provide a better basis for inference about mechanistic
interaction.  The proposed method does however, rely on the assumption
that the mechanism studied is, at some deep level, deterministic ---
which is by no means universally appropriate, as shown by \cite{dawid2004}.  This
assumption can, however, be defensible in some fields of application,
and our choice of an illustrative study in molecular medicine reflects
such concerns.


\newpar Finally, we would re-iterate that, unlike previous approaches
to the problem, the proposed method avoids artificial mathematical
constructs based on a potential response paradigm of statistical
causality.  While some of our tests are mathematically similar to
previously proposed tests based on ``principal stratum'' arguments,
our tests differ in that we insist the context variable $V$ be both
real and relevant. Although $V$ may be wholly or partly unobserved, it
is important in our method that it be, in principle at least,
observable, and that its relationships with the remaining variables in
the problem explicitly represented in the causal model.  With the aid
of study examples, we have shown that such an exercise is necessary to
differentiate situations in which the method is applicable from
situations in which it is not.

{\small \bibliographystyle{plain} \bibliography{REFERENCES}}

\vspace{1.5cm}

\begin{table}[ht]
  \begin{center}
    \begin{tabular}{rrrrr}
      & Estimate & Std. Error & z value & $p$-value \\
      \hline
      $\alpha$ (intercept) & -2.33 & 0.5 & -4.66 & 3$e^{-6}$ \\
      $\phi_{G = 1}$ (wild-type rs1333040 homozygous)  & -0.06 & 0.19 & -0.31 & 0.7 \\
      $\phi_{S = 0}$ (no statin treatment)& 1.41 & 0.24 & 5.85 & 4e-09 \\
      $\delta_T$ (linear effect of calendar year - 1970)& -0.02 & 0.017 & -1.52 & 0.12 \\
      $\gamma_{(S \ugualino 0) \times (G \ugualino 1)}$ & -1.0 & 0.33 & -3.0 & 0.002 \\
    \end{tabular}
  \end{center}
  \begin{flushleft}
\caption{\small Parameter estimates from a linear-odds regression
of the prospective binary endpoint in our illustrative
study (re-infarction within six years from the index
infarction) upon variables $S$ (the statin treatment indicator) and
$G$ (a function of the genotype at SNP rs1333040). Variable
$G$ is coded to take value 1 if the individual carries
two copies of the most frequent allele at single
nucleotide polymorphism rs1333040. This table reports
estimates for the parameters of the regression model,
as obtained from an analysis of 1200 subjects who were
hospitalized on the basis of a myocardial infarction
between 40 and 45 years of age, and were found at that
point to have hypercholesterolemia.
These estimates
suggest that, in patients with
hypercholesterolemia, statins decrease the risk of re-infarction regardless
of the rs1333040 genotype ($G$), although their effect is stronger in
patients with $G = 0$. At highest risk are those
hypercholesterolemic patients with $G = 0$ who do
not receive statins. Because the quantity $\gamma-\alpha$ is
significantly greater than zero, we deduce that $G$ and $S$
interfere with each other (and hence strongly coact) to
produce re-infarction.
\label{Table 1}}
\end{flushleft}
\label{tab:1}
\end{table}

\begin{table}[ht]
  \begin{center}
    \begin{tabular}{lllll}
      \hline
      && Standard&&\\
      & Estimate
      &\begin{minipage}{1.2cm}\begin{center}Error\end{center}\end{minipage}&
      z value & $p$-value \\
      \hline
      Intercept ($\alpha$)      &0.25&0.013&18.7& $<$ 2$e^{-16}$\\
      smoker              &1.46&0.044&33.1& $<$ 2$e^{-16}$ \\
      rare rs4620585 homozygous (RRH)&0.07&0.056&1.2&0.19\\
      smoker $\times$ RRH  ($\gamma$) &0.9&0.22&4.09&4$e^{-5}$
    \end{tabular}
  \end{center}
  \begin{flushleft}
\caption{\small Parameter estimates from a linear-odds regression of the
early MI indicator upon ``smoking habit'', obtained by dichotomizing
an original ``Daily number of cigarettes'' variable and genotype
at SNP rs4620585, based on the set of cases and controls of
our Illustrative study.
\label{Table 3}}
\end{flushleft}
\label{tab:3}
\end{table}

\end{document}